\documentclass[11pt, draftcls, onecolumn]{IEEEtran}
\usepackage{amsmath}
\usepackage{amsthm}
\usepackage{amssymb}

\usepackage{ifthen}
\newboolean{showcomments}
\setboolean{showcomments}{true}
\newcommand{\jk}[1]{  \ifthenelse{\boolean{showcomments}}
{ \textcolor{red}{(JK says:  #1)}} {}  }
\newcommand{\ignore}[1]{}

\usepackage{tikz}
\usepackage{pgfplots}
\definecolor{bblue}{HTML}{4F81BD}
\definecolor{rred}{HTML}{C0504D}
\definecolor{ggreen}{HTML}{9BBB59}
\definecolor{ppurple}{HTML}{9F4C7C}

\usetikzlibrary{calc}

\newtheorem{lemma}{Lemma}

\newtheorem{theorem}{Theorem}
\newtheorem{remark}[theorem]{Remark}

\def\bb0{{\mathbb{0}}}

\def\bb{{\mathbf{b}}}

\def\b0{{\mathbf{0}}}

\def\opt{\mathsf{OPT}}

\def\b1{{\mathbf{1}}}

\def\cA{\mathcal{A}}

\def\cJ{\mathcal{J}}

\def\sf0{{\mathsf{0}}}

\def\nn{\nonumber}

\newlength{\figurewidth}

\newcounter{one}
\newcounter{two}
\author{\parbox{5 in}{ \centering Rahul Vaze\\
        School of Technology and Computer Science \\
        Tata Institute of Fundamental Research\\
      	Mumbai, India\\
	rahul.vaze@gmail.com}}
\begin{document}
\title{ Scheduling for Multi-Phase Parallelizable Jobs}

%
%

\maketitle

\begin{abstract}
With multiple identical unit speed servers, the online problem of scheduling jobs that migrate between two phases, limitedly parallelizable or completely sequential, and choosing their respective speeds to minimize the total flow time  is considered. 
In the limited parallelizable regime, allocating $k$ servers to a job, the speed extracted is $k^{1/\alpha}, \alpha>1$, a sub-linear, concave speedup function, while in the sequential phase, a job can be 
processed by at most one server with a maximum speed of unity. 
A LCFS based algorithm is proposed for scheduling jobs which always assigns equal speed to the jobs that are in the same phase (limitedly parallelizable/sequential), and is shown to have a constant (dependent only on $\alpha > 1$) competitive ratio. For the special case when all 
jobs are available beforehand, improved competitive ratio is obtained.
\end{abstract}

\section{Introduction}
In the presence of multiple servers, how to schedule {\it parallelizable} jobs to minimize the sum of their response times (called the flow time) is an incredibly important 
and analytically challenging problem, e.g. in large data centers. With multiple servers, the parallelizability of job is captured by the total speed assigned to it when processed by multiple servers  simultaneously. 
Let the total number of servers be $N$, where each server can operate at the maximum speed of unity. Then, typically \cite{berg2017towards, berg2019hesrpt, im2016competitively, edmonds2000scheduling, edmonds2009scalably, agrawal2016scheduling,  verma2015large, lin2018model, vazeperf2021}, if $k\le N$ is the number of servers assigned to a job, the resulting speed obtained is $s(k) = k^{1/\alpha}$. Depending on $\alpha$ (called the speed-up exponent), i) if $\alpha=1$, the job
is called fully parallelizable, otherwise if $\alpha >1$, its called limitedly parallelizable, while if $\alpha=\infty$ for $\forall \ k>1$ and $\alpha=1$ for $k\le 1$, it is called sequential. 

In most practical settings \cite{d1,o2009star,tallent2009effective,nguyen1996using,shvachko2010hadoop} each job does not necessarily have a  single phase of parallelizability, but migrates between different phases at different times during its execution. For example in a MapReduce framework \cite{shvachko2010hadoop}, initially, jobs have full/limited parallelizability, while in the concluding stages they become sequential. 
Given practical considerations as described in detail in \cite{berg2021case}, it is reasonable to consider the case of jobs having either limited parallelizability \cite{berg2019hesrpt} (called elastic phase), or are sequential (called in-elastic phase), where $2\le \alpha\le 3$ is the most relevant regime for limited parallelizability. 

Thus, in this paper, we consider the online problem of scheduling jobs and how many servers to allocate to each job being processed to minimize the flow time, where each job has two possible types of phases of parallelizability, either elastic or in-elastic, and where jobs arrive at arbitrary times, have arbitrary number of elastic and in-elastic phases, 
and have arbitrary job sizes for each phase. 
 To quantify the performance of an online algorithm, we consider the metric of  competitive ratio, that is defined as the ratio of the flow time of the online algorithm and the optimal offline algorithm $\opt$ (that knows the entire input sequence in advance) maximized over all possible inputs (worst case).

\subsection{Prior Work}
\subsubsection{Single Phase}

With limited parallelizability, the single phase scheduling problem of finding how many servers to allocate to each job that minimizes the flow time is challenging, and has been an object of immense interest \cite{berg2017towards, berg2019hesrpt, im2016competitively, edmonds2000scheduling, edmonds2009scalably, agrawal2016scheduling, vazeperf2021}. 
With limited parallelizability, the single phase scheduling problem has been considered for two models i) the combinatorial discrete allocation model \cite{im2016competitively}, where an integer number of servers are assigned to any job, and ii) the continuous allocation model \cite{edmonds2000scheduling, edmonds2009scalably, agrawal2016scheduling, berg2017towards, berg2019hesrpt, vazeperf2021}, that treats the $N$ servers as a single resource block which can be partitioned into any size and assigned to any job. 
In the continuous allocation model, for the online case where jobs arrive over time, \cite{vazeperf2021} proposed a constant competitive algorithm that only depends on the exponent $\alpha$, while an optimal algorithm has been derived in \cite{berg2019hesrpt} when all jobs are available at time $0$. In practice, some of the methods for server allocation include packing based \cite{verma2015large}, and resource reservation algorithms \cite{ren2016clairvoyant}. Heuristic policies with only numerical performance analysis  can be found in \cite{lin2018model}.

\subsubsection{Multiple Phases}
The multiple phase scheduling problem has primarily been considered in the continuous allocation model  \cite{edmonds2000scheduling, edmonds2009scalably, agrawal2016scheduling}, where there are arbitrary number of phases with arbitrary speed-up exponents $\alpha$ for each phase.
 In this line of work,
 mostly the non-clairvoyant setting (the algorithm is not  aware of the remaining size of the jobs or the exponent $\alpha$ of the current/future phases), with few exceptions where clairvoyant setting has been studied \cite{turek1994scheduling1,turek1994scheduling2}. The competitive ratio of any non-clairvoyant online algorithm (both deterministic and randomized) is known to be at least $\sqrt{n}$ ($n$ is the total number of jobs) \cite{edmonds2000scheduling},  when there are arbitrary number of phases with different exponents $\alpha$.
 
In light of the lower bound, resource augmentation is considered, where an algorithm is allowed more resources than the optimal offline algorithm. Algorithms with constant competitive ratios have been derived as a function of the resource augmentation factor \cite{edmonds2000scheduling, edmonds2009scalably}.  In particular, algorithm \textsf{EQUI} that assigns equal speed to all jobs (without knowing even the current phase index for each job) has a constant competitive ratio when given double the number of servers compared to the $\opt$ \cite{edmonds2000scheduling}. A more refined competitive ratio result with resource augmentation was derived in \cite{edmonds2009scalably}.
 Surprisingly, for the special case, where all phases are strictly $\rho$ sub-linear for any $\rho>0$, where the speed function $s(k)$ (speed 
assigned to job when allocated $k$ servers) satisfies the relation $\frac{s(k_2)}{s(k_1)} \le \left(\frac{k_2}{k_1}\right)^{1-\rho}$ whenever $k_1\le k_2$, \textsf{EQUI} has a competitive ratio  of 
$2^{1/\rho}$ against a clairvoyant optimal offline algorithm without any resource augmentation \cite{edmonds2000scheduling}. Notably, the in-elastic phase considered in this paper is not strictly $\rho$ sub-linear.


From a practical point of view, the two phase problem is more relevant, and for which 
heuristic policies, e.g., the phase-aware FCFS \cite{phaseawareFCFS} that schedules jobs in their arrival order, while assigning at most speed $1$ to a job that is in its in-elastic phase, have been proposed.
Some partial results have been derived in \cite{berg2020optimal} for the two-phase scheduling problem. In very recent work, \cite{berg2021case} characterized an optimal scheduling policy, for the two-phase scheduling problem as studied in this paper, however, with two strong assumptions, i) 
the size of jobs in the elastic and in-elastic phases are exponentially distributed with the same parameters for all jobs, and are independent of each other, and ii) the job always completes when it is in its in-elastic phase. 
We avoid all these assumptions in this paper, by letting the job sizes in each phase to be arbitrary, and the first and the last phase of a job can either be elastic or in-elastic.

\subsection{Our contributions} 
For the two-phase scheduling problem, we propose an algorithm called \textsc{Fractional-LCFS} that processes a fraction of the outstanding jobs that have arrived most recently, and a subset of inelastic jobs, where each type of scheduled job is executed with equal speed. The exact choice is more refined and detailed in Section \ref{sec:algo}.  The algorithm is semi non-clairvoyant that disregards the remaining job sizes of all remaining phases (even though they are known), and only uses the information about the current phase each job is in. 

The choice of which jobs to process by the algorithm is defined by the number of jobs in each of the two-phases. Compared to the algorithm \cite{berg2021case} that always prioritises jobs that are in their in-elastic phases, our algorithm prioritises jobs that are in their in-elastic phase only when there are sufficiently many of them and the total number of jobs is less compared to the 
total number of servers. 

We show that \textsc{Fractional-LCFS} has a constant competitive ratio (derived in Theorem \ref{thm:flowtimeplusenergyonline}) that depends only on the speed-up exponent $\alpha > 1$ and not on system parameters such as the total number of jobs, and their respective sizes, and the number of servers. This result overcomes fundamental 
challenge left open in the literature for the considered problem, where speed augmentation was needed to prove constant competitiveness \cite{edmonds2000scheduling}. 
It is worth mentioning that we do not get any meaningful competitive ratio when $\alpha=1$ (fully-parallizable jobs), since 
for this case, a lower bound of $n^{1/3}$ ($n$ is the total number of jobs) on the competitive ratio is known \cite{motwani1994nonclairvoyant} for any deterministic algorithm that is unaware of the remaining sizes of the jobs, similar to the algorithm  \textsc{Fractional-LCFS}. 

We also consider the simpler setting where all jobs are available at time $0$. Similar to the online jobs arrival case, in this case also, we propose an algorithm that makes three different choices on which jobs to schedule depending on the number of jobs in the system and the number of servers. Moreover, it assigns
equal speed to all jobs that are being processed that belong to the same phase. Compared to the online jobs arrival case, we get a significantly improved competitive 
ratio bound in this simpler case provided in Theorem \ref{thm:flowtimetimezero}. It is worth recalling that  an optimal algorithm for the single phase problem where all jobs are available at time $0$ has been derived in \cite{berg2019hesrpt}, however, no such result is known for the two-phase problem.


In addition to the analytical results, we also present average-case simulation results to illustrate the actual performance of the proposed algorithm. We compare the performance of our proposed algorithm with $\mathsf{EQUI}$, the inelastic first algorithm \cite{berg2021case}, as well as the phase aware FCFS  \cite{phaseawareFCFS}, and observe that 
the performance of our algorithm is comparable or better than  $\mathsf{EQUI}$ and the inelastic first algorithm, while outperforming phase aware FCFS always.


\section{System Model}\label{sec:sysmodel}
Let there be $N$ parallel and identical servers, each with speed $1$. The set of jobs is denoted by $\cJ$, where a job $j \in \cJ$ arrives at time $a_j$. Similar to \cite{edmonds2000scheduling, berg2019hesrpt}, we consider the continuous allocation model, where $N$ is treated as a single resource block which can be divided into chunks of arbitrary sizes and allocated to different jobs.

Each job $j$ at any time can be in one of two phases, called elastic or in-elastic. The sizes of job $j$ in the $a^{th}$ elastic and $b^{th}$ in-elastic phase are $w_{je}^a$ and $w_{j\iota}^b$, respectively. Moreover, let $A_j$ and $B_j$ be the total number of elastic and in-elastic (interleaved) phases required for each job, respectively. The first and the last phase of any job  can be either of the two phases. We consider the online setting, where an algorithm has only causal information about jobs, i.e. any job's phases and their respective sizes are revealed only once it arrives.

In the elastic phase, any job is parallelizable with concave speedup, i.e.,  if job~$j$ is allotted 
$k_j(t)$ number of
servers at time $t$, then the service rate experienced by
job~$j$ at time $t$ is $s_j(t) = P(k_j(t)) = k_j(t)^{1/\alpha},$ where $\alpha > 1$. Note that with the continuous 
allocation model, it is possible that $k_j(t) < 1$. Following \cite{berg2019hesrpt, vazeperf2021, berg2021case}, however, we let $s_j(t) =  k_j(t)^{1/\alpha}$ even when $k_j(t) < 1$. 

In the in-elastic phase, each job can be processed by at most one server, and equivalently can be processed at speed of at most $1$. 
Moreover, for any job $j$, it transitions from the elastic to in-elastic phase or vice versa only when its total work $w_{je}^a$ or $w_{j\iota}^b$ in the current phase is complete.

A job $j$ is defined to be complete at time $d_j$, if $d_j$ is the earliest time at which total $\sum_{a\le A_j} w^a_{je} + \sum_{b \le B_j} w^b_{j\iota}$ amount of work has been completed for job $j$, and the objective is to minimize the flow time 
\begin{equation}\label{eq:flowtime}
\min F = \sum_{j\in \cJ} (d_j-a_j) =  \int n(t) dt \ \text{s.t.} \ \sum_{j= 1}^{A(t)} P(s_j(t)) \leq N,
\end{equation}
 where $n(t)$ is the number of outstanding jobs at time $t$, and  $A(t)$ is the set of jobs that are being processed at time $t$.

 Compared to our general system model, recently in \cite{berg2021case}, a three state Markov chain was considered for phase transitions as shown in Fig. \ref{fig:BD}, where each job arrives in either 
 the elastic or the in-elastic phase, and transitions between the two phases at fixed rates, and always exits from the in-elastic phase. Because of these strong assumptions, \cite{berg2021case} was able to identify an optimal policy that always prioritizes the jobs for scheduling that are in their in-elastic phases. With the general system model, this is no longer true, and in Section \ref{sec:algo}, we present a different algorithm and show that its competitive ratio is a constant.
 
 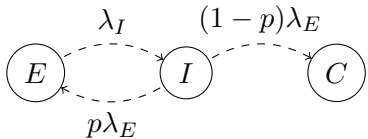
\begin{figure}[h]
\centering
\begin{tikzpicture}[level/.style={sibling distance=50mm/#1}]
\node[circle,draw] (s) at (0,0) {$E$};
\node[circle,draw] (a) at (2,0) {$I$};
\node[circle,draw] (b) at (4,0) {$C$};

\draw[bend right, dashed,<-]  (s) to node [below] {$ p \lambda_E$} (a);
\draw[bend right, dashed,<-]  (a) to node [above] {$\lambda_I$} (s);

\draw[bend right, dashed,<-]  (b) to node [above] {$(1-p) \lambda_E$} (a);




\end{tikzpicture}
\caption{Birth death chain for phase evolution, where $I, E, C$ represent the in-elastic, elastic and completion phases, respectively, and $\lambda$'s are the transition rates, and $0 < p < 1$ is constant.}
\label{fig:BD}
\end{figure}
\section{Metric}
We represent the optimal offline algorithm (that knows the entire job arrival sequence including the number of phases, and the respective sizes of jobs in each phase,  in advance) as $\opt$. 
Let $n(t)$ ($n_o(t)$)  be the number of outstanding jobs with an online algorithm $\cA$ ($\opt$). 
For Problem \eqref{eq:flowtime}, we will consider the metric of competitive ratio which for an online algorithm $\cA$ is defined as 
\begin{equation}\label{defn:cr}\mu_\cA  = \max_\sigma \frac{\int n(t)  dt}{\int n_o(t) dt},
\end{equation}
where $\sigma$ is the input sequence consisting of jobs set $\cJ$. 

We will propose an online algorithm $\cA$, and bound $\mu_\cA \le \kappa$, by showing that for each time instant $t$
\begin{align}\label{eq:runcond}
  n(t) +  d\Phi(t)/dt & \le \kappa  n_o(t),
\end{align}
where $\Phi(t)$ is some function called the {\bf potential function} that satisfies the boundary conditions:
\begin{itemize}
\item $\Phi(t) = 0$ initially before all job arrivals and $\Phi(\infty) = 0$.
\item $\Phi(t)$ does not increase on any job arrival or job departure with the algorithm or the $\opt$.
\end{itemize}
Integrating \eqref{eq:runcond} over time, implies that the competitive ratio of $\cA$ is at most $\kappa$.

\section{Algorithm \textsc{Fractional-LCFS}}\label{sec:algo}
In this section, we propose an algorithm that is semi non-clairvoyant, that disregards the information about the remaining job sizes of all the remaining phases, and only exploits the binary information about a job being in the elastic or the in-elastic phase, which will be compared against a clairvoyant optimal offline algorithm in terms of the competitive ratio.
At time $t$, let the outstanding number of jobs in the system be $n(t)$, and $n_\iota(t)$ be the number of jobs that are in their in-elastic phase. 
Thus, $n(t) = n_e(t)+ n_\iota(t)$, where $n_e(t)$ is the number of jobs that are in their elastic phase.

{\bf Scheduling and speed selection:} 
Let $\beta, \theta$ be constants with $0 < \theta < \beta < 1$.

Case I $\frac{N}{\beta n(t)} \le 1$: Process the $\beta n(t)$ jobs that have arrived {\bf most recently} without distinguishing between jobs that are in their elastic or in-elastic phase. \footnote{If $\beta n(t)$ is fractional, then we mean $\lceil \beta n(t) \rceil$.}  
{\bf Speed:} Each of the $\beta n(t)$ jobs are processed at equal speed 
\begin{align}\label{}
  s(t) & = P\left(\frac{N}{ \beta n(t)}\right).
\end{align}

Case II $\frac{N}{\beta n(t)} > 1$: 
IIa: If $n_\iota(t) \ge \theta n(t)$\footnote{If $\theta n(t)$ is fractional, then we mean $\lceil \theta n(t) \rceil$.} then process any $\min\{n_\iota(t),N\}$ jobs\footnote{If $N=\min\{n_\iota(t),N\}$, then pick any $N$ jobs out of total $n_\iota(t)$ jobs.}  that are in their in-elastic phase, and among the $n_e(t)$ jobs  that are in their elastic phase, 
process the $\beta n_e(t)$ that have arrived {\bf most recently}.
{\bf Speed:} 
\begin{align}\label{}
  s(t) & = \begin{cases} 1 & \text{for each of}  \min\{n_\iota(t),N\} \ \text{jobs}, \\ 
  P\left(\frac{N-\min\{n_\iota(t),N\}}{\beta n_e(t)}\right) & \text{for each of} \ \beta n_e(t) \ \text{jobs}.\end{cases}
\end{align}
IIb:If $n_\iota(t) < \theta n(t)$ Among the $\beta n(t)$ jobs that have arrived {\bf most recently},  process all the jobs that are in their elastic phases with equal speed
\begin{align}\label{}
  s(t) & =  P\left(\frac{N}{ \beta n(t)}\right). 
\end{align}
Note that in this subcase, the total speed  
constraint of $\sum_{j= 1}^{A(t)} P^{-1}(s_j(t)) \leq N$ need not be tight. Thus, for a practical implementation, few more jobs can be processed, however, 
that will not change the analysis.

By its very definition, algorithm \textsc{Fractional-LCFS} satisfies the total speed  
constraint of $$\sum_{j= 1}^{A(t)} P^{-1}(s_j(t)) \leq N,$$ as well as the speed constraint of unity for any job that is in its in-elastic phase.

The main result of this paper is as follows.
\begin{theorem}\label{thm:flowtimeplusenergyonline}
  For any $\alpha > 1$, there exists a $0 < \theta < \beta < 1$, such that the competitive ratio of algorithm \textsc{Fractional-LCFS} for Problem \eqref{eq:flowtime} is a constant (depends only on $\alpha$) and is independent of the number of jobs, their sizes, and the number of servers $N$. The exact competitive ratio 
  expression is provided in \eqref{eq:finalbound}, and using which for example in case of $\alpha=2$, we get the competitive ratio bound of $636$, choosing $\beta = \frac{1}{6}$, and $\theta= \frac{1}{72}$.\end{theorem}

For each value of $\alpha>1$, how to choose $\beta,\theta$ such that the competitive ratio remains a constant is discussed in Remark \ref{rem:choiceofbeta}. We are prescribing only one potential choice of parameters $ \theta, \beta$ that is sufficient to make the competitive ratio as a constant, however, there is scope for choosing the parameters $ \theta, \beta$ so as to minimize the competitive ratio. Doing so 
analytically, however, remains a challenge, while easy being numerically.

\begin{remark}\label{}Our result does not result in any meaningful bound for $\alpha=1$ as expected, since the  lower bound of $n^{1/3}$ ($n$ is the total number of jobs) on the competitive ratio is known \cite{motwani1994nonclairvoyant} for $\alpha=1$ for any non-clairvoyant algorithm, as is the \textsc{Fractional-LCFS} algorithm. 
\end{remark}
\begin{remark}\label{}
  It is worth noting that the competitive ratio bound in Theorem \ref{thm:flowtimeplusenergyonline} increases as $\alpha \rightarrow 1$. The main intuition for this is that we are considering the worst case input, which includes the case where jobs have no in-elastic phases, for which as $\alpha \rightarrow 1$, SRPT is an optimal algorithm that processes only one job 
  with the least remaining size on all servers. In contrast, with \textsc{Fractional-LCFS}, potentially a large number of jobs are parallely processed with equal speed for all values of $\alpha$. 
\end{remark}

{\it Discussion:} 
Theorem \ref{thm:flowtimeplusenergyonline} shows that a simple LCFS algorithm that processes a fraction of the most recently arrived outstanding jobs, that is 
not even aware of the remaining job size (of any remaining phase) and that uses equal 
speed for jobs that are in the same phase, is constant competitive, i.e., independent 
of input parameters: number of jobs and their sizes, and the number of servers, and 
only depends on the speedup exponent $\alpha$. Even though the derived competitive ratio bound appears large, it overcomes an
old technical hurdle of it being independent of system parameters.
In prior work, either speed augmentation \cite{edmonds2000scheduling} was shown to be necessary to get similar
constant competitive ratio results, or somewhat simplistic input model 
had to be considered \cite{berg2021case}. Moreover, given the very nature of the competitive ratio metric being a multiplicative penalty, a large competitive ratio per se is not limiting, as long as it does not scale with system parameters.

The intuition as to why a fractional LCFS algorithm should perform well  is similar to that of 
the SRPT (shortest remaining processing time) algorithm that requires the knowledge of remaining job sizes. 
SRPT minimizes the number of outstanding jobs (that controls the flow time) knowing the jobs sizes, by keeping 
shorter jobs in the system for less time. Fractional LCFS on the other hand, without using the remaining job size information, 
processes a fraction of the most recently
arrived jobs, and tries to keep  longer jobs stay in the system for long, thus `effectively' prioritizing short jobs. It is easy to construct 
`bad' input sequences where this is not the case, but roughly that is what one should expect.

The speed choice made by the proposed algorithm is primarily dictated by the constructed potential function and the unity speed constraint for the in-elastic phase, so 
that the overall drift (derivative) of the potential function is sufficiently large. In particular, the speed chosen for 
jobs that are in the same phase is always identical. The intuition for the equal speed choice can be 
borrowed from \cite{edmonds2000scheduling}, that explains that if an algorithm choosing equal speed has more number of 
outstanding jobs than the $\opt$, then 
progressively, it allocates fewer servers to each job and since $\alpha >1$, it improves the utilization 
of servers. Since we also have jobs that are in their in-elastic phase, this is not precisely correct, however, provides partial 
explanation.


After dealing with the setting where jobs arrive at arbitrary times, next, we consider the simpler case when all jobs are available at time $0$ and get a better competitive ratio 
guarantee.

\section{All jobs available at time $0$} In this section, except for all jobs arriving at time $0$, everything is identical to the system model described in Section \ref{sec:sysmodel}.

\subsection{Algorithm \textsc{PA-EQUI}}
{\bf Scheduling and speed selection:} 

Case I $\frac{N}{n(t)} \le 1$: Process all the $n(t)$ (number of outstanding) jobs, without distinguishing between jobs that are in their elastic or in-elastic phase, with equal speed 
\begin{align}\label{}
  s(t) & = P\left(\frac{N}{n(t)}\right).
\end{align}

Case II $\frac{N}{n(t)} > 1$: 
IIa: For a constant $0 < \delta < 1$, if $n_\iota(t) \ge \delta n(t)$ then process all $n_\iota(t)$ jobs  that are in their in-elastic phase dedicatedly in one server with unit speed, while process the remaining  $n_e(t)$ jobs that are in their elastic phase, each with speed $P\left(\frac{N-n_\iota(t)}{ n_e(t)}\right)$.

IIb:If $n_\iota(t) < \delta n(t)$ Process all the $n_e(t)$ jobs that are in their elastic phase, each with equal speed
\begin{align}\label{}
  s(t) & =  P\left(\frac{N}{n_e(t)}\right). 
\end{align}

We name this algorithm \textsc{PA-EQUI}, since it allocates equal speed to all jobs that belong to the same phase. In contrast, \textsc{EQUI} studied in \cite{edmonds2000scheduling,edmonds2009scalably}
 is \textsc{Blind-EQUI}, since it is unaware which jobs belong to which phase, and wastes speed. 
By its very definition, algorithm \textsc{PA-EQUI} satisfies the total speed  
constraint of $\sum_{j= 1}^{A(t)} P^{-1}(s_j(t)) \leq N$, as well as the speed constraint of unity for any job that is in its in-elastic phase.

The main result of this section is as follows.
\begin{theorem}\label{thm:flowtimetimezero}
  For $\alpha >1$, the competitive ratio of \textsc{PA-EQUI} for Problem \eqref{eq:flowtime} when all jobs are available at time $0$, is at most $$\mu(\alpha) = \frac{1}{\alpha(1-\delta)-1}\left[\frac{\alpha(1-\delta)}{\delta} + \frac{\alpha(1-\delta)+\delta}{1-\delta}\right],$$
  where $\delta$ is the parameter to be chosen.
   For $\alpha=2$, choosing $\delta=\frac{1}{4}, \mu(2) = 50/3$. 
  Moreover, $\mu(\alpha)$ is a decreasing function of $\alpha > 1$ for an appropriate choice of $\delta$.
\end{theorem}

Thus, compared to the online job arrivals case (Theorem \ref{thm:flowtimeplusenergyonline}) where the competitive ratio for $\alpha=2$ is $636$,
there is a significant improvement in the competitive ratio when all jobs are available at time $0$. Similar conclusion can be drawn for other values of $\alpha$ also.
To prove Theorem \ref{thm:flowtimetimezero}, similar to the previous section, we consider the following potential function, and show that  \eqref{eq:runcond} holds for 
a particular value of $\kappa$.

{\bf Potential Function} 
At time $t$, let $A(t)$ be the set of unfinished jobs with \textsc{PA-EQUI}  where $n(t) = |A(t)|$, and for the $i^{th}$ job,  $i\in A(t)$, let $q_i(t)$ be its remaining (sum of remaining sizes of all the remaining phases) size. Then  
\begin{align}\label{}
 n^i(t,q) & = \begin{cases} 1 & \text{for} \  q\le q_i(t), \\
 0 & \text{otherwise.}
 \end{cases}
\end{align} 
Similarly, let $n_o(t)$ be the number of unfinished jobs with the $\opt$, and the corresponding quantity to $n^i(t,q)$ for the $i^{th}$ job with the $\opt$, be denoted by $n^i_o(t,q)$.

Consider the potential function
\begin{equation}\label{defn:phispecialFT}
\Phi^{sf}(t) = c_1\Phi^{sf}_1(t) + c_2\Phi_2(t),
\end{equation}
where 
\begin{equation}\label{defn:phi1specialFT} 
\Phi^{sf}_{1}(t) =  P\left(\frac{n(t)}{N}\right) \left(\sum_{i\in A(t)}\int_{0}^\infty ( n^i(t,q) - n^i_o(t,q))^+ dq\right),
\end{equation}
where $c_1$, and $c_2$ are constants to be chosen later, and $(x)^+ = \max\{0,x\}$, and $\Phi_2(t)$ is as defined in \eqref{phi2}.

Clearly, $\Phi^{sf}(t)$ satisfies the first boundary condition. 
Since all jobs are available at time $0$, which is equivalent to all arrivals happening at time $t=0$, both $\Phi^{sf}_{1}(t)=0$ and $\Phi_2(t)=0$ for $t=0$. Thus, to check whether $\Phi^{sf}(t)$ satisfies the second boundary condition, 
we only need to check whether $\Phi^{sf}(t)$ increases on a departure of a job with either the \textsc{PA-EQUI} or the $\opt$.
\begin{lemma}\label{lem:jumpequi}
Potential function $\Phi^{sf}(t)$ \eqref{defn:phispecialFT} does not increase on a departure of a job with either the \textsc{PA-EQUI} or the $\opt$.
\end{lemma}

The proof of Lemma \ref{lem:jumpequi} is provided in Appendix \ref{app3}.
Next, we characterize the drift $d \Phi^{sf}(t)/dt$.
\begin{lemma}\label{lem:driftphisopt}
Because of the processing by the $\opt$,  the change in the potential function \eqref{defn:phispecialFT} is
\begin{align*}
d \Phi^{sf}(t)/dt & \le   c_1 \left(\frac{1}{\alpha}\right) n(t) + c_1\left(1-\frac{1}{\alpha}\right) n_o(t) + c_2 n_o(t).
\end{align*}
\end{lemma}

The proof of Lemma \ref{lem:driftphisopt} is provided in Appendix \ref{app4}.
\begin{lemma}\label{lem:driftphisalg}
Because of the processing by the algorithm \textsc{PA-EQUI}, the change in the potential function \eqref{defn:phispecialFT} is
 \begin{align} 
 d\Phi^{sf}(t)/dt & \le \begin{cases} -c_1(\max\{n(t)-n_o(t),0\} )  &  \text{if} \ \frac{N}{ n(t)} \le 1, \\
  -c_2 n_\iota(t) &  \text{if} \ \frac{N}{ n(t)} > 1 \ \text{and} \ n_\iota(t) \ge \delta n(t),\\
  - c_1 (\max\{n_e(t)-n_o(t),0\} )& \text{if} \ \frac{N}{ n(t)} > 1 \ \text{and} \ n_e(t) \ge (1-\delta) n(t).
   \end{cases}
 \end{align}
\end{lemma}
The proof of Lemma \ref{lem:driftphisalg} is provided in Appendix \ref{app5}.
With these preliminaries, we are ready to prove  Theorem \ref{thm:flowtimetimezero}.
\begin{proof}[Proof of Theorem \ref{thm:flowtimetimezero}]
Case I : $\max\{n(t)-n_o(t),0\}  =0.$ In this case, we only count the $\opt$'s contribution to $d\Phi^{sf}(t)/dt$ from Lemma \ref{lem:driftphisopt}, since $d\Phi^{sf}(t)/dt$ because 
of \textsc{PA-EQUI}'s processing is always non-positive. Thus, we can write   \eqref{eq:runcond}, as 
\begin{align}\nn
n(t)  +   d\Phi^{sf}(t)/dt& \le n(t)   + c_1 \left(\frac{1}{\alpha}\right) n(t) + c_1\left(1-\frac{1}{\alpha}\right) n_o(t) + c_2 n_o(t), \\ \nn
& \le n_o(t)(1+c_1+c_2), \end{align}
where the second inequality follows since $n(t) \le n_o(t)$.

Case II: $n_o(t)> 0$,  $\max\{n(t)-n_o(t),0\} = n(t)-n_o(t)$, and $n(t) \ge N$.
Using Lemma \ref{lem:driftphisopt} and \ref{lem:driftphisalg}, we can write \eqref{eq:runcond},
\begin{align}\nn
 n(t)  +   d\Phi^{sf}(t)/dt& \le n(t)   + c_1 \left(\frac{1}{\alpha}\right) n(t) + c_1\left(1-\frac{1}{\alpha}\right) n_o(t) + c_2 n_o(t) -c_1(\max\{n(t)-n_o(t),0\} ), \\  \nn
  & \le n(t)\left(1+ c_1\left(\frac{1}{\alpha}\right)-c_1 \right) + n_o(t)\left(c_1\left(1+\left(1-\frac{1}{\alpha}\right)\right)+c_2\right),\\  \label{eq:runcond2}
  & \le n_o(t)\left(c_1\left(1+\left(1-\frac{1}{\alpha}\right)\right)+c_2\right),
\end{align}
for $ c_1\ge 1/\left(1-\left(\frac{1}{\alpha}\right)\right)$.

Case II: $n_o(t)> 0$,  $\max\{n(t)-n_o(t),0\} = n(t)-n_o(t)$,  $n(t) < N$  and $n_\iota(t) \ge \delta n(t)$.
Using Lemma \ref{lem:driftphisopt} and \ref{lem:driftphisalg}, we can write \eqref{eq:runcond},
\begin{align}\nn
 n(t)  +   d\Phi^{sf}(t)/dt& \le n(t)   + c_1 \left(\frac{1}{\alpha}\right) n(t) + c_1\left(1-\frac{1}{\alpha}\right) n_o(t) + c_2 n_o(t) -c_2 n_\iota(t), \\  \nn
  & \stackrel{(a)}\le n(t)\left(1+ c_1\left(\frac{1}{\alpha}\right)- \delta c_2 \right) + n_o(t)\left(c_1\left(1-\frac{1}{\alpha}\right)+c_2\right),\\  \label{eq:runcond2}
  & \le n_o(t)\left(c_1\left(1-\frac{1}{\alpha}\right)+c_2\right),
\end{align}
where $(a)$ follows since $n_\iota(t) \ge \delta n(t)$, and the final inequality follows
for $\delta c_2 \ge 1+\frac{c_1}{\alpha}$.

Case III: $n_o(t)> 0$,  $\max\{n(t)-n_o(t),0\} = n(t)-n_o(t)$,  $n(t) < N$  and $n_e(t) \ge (1-\delta) n(t)$.

Case IIIa: $\max\{n_e(t)-n_o(t),0\}  =0.$
 In this case, we only count the $\opt$'s contribution to $d\Phi^{sf}(t)/dt$ from Lemma \ref{lem:driftphisopt} since $d\Phi^{sf}(t)/dt$ because 
of \textsc{PA-EQUI} is always non-positive. Thus, we can write   \eqref{eq:runcond}, as 
\begin{align}\nn
n(t)  +   d\Phi^{sf}(t)/dt& \le n(t)   + c_1 \left(\frac{1}{\alpha}\right) n(t) + c_1\left(1-\frac{1}{\alpha}\right) n_o(t) + c_2 n_o(t), \\ \nn
& \le n_o(t)\left(\frac{1}{(1-\delta)}+c_1\left(1+\frac{\frac{1}{1-\delta}-1}{\alpha}\right)+c_2\right), \end{align}
where the second inequality follows since $n(t) \le \frac{n_o(t)}{1-\delta}$ as $n_e(t) \ge (1-\delta) n(t)$ and $\max\{n_e(t)-n_o(t),0\}  =0$.

Case IIIb: 
$\max\{n_e(t)-n_o(t),0\}  =n_e(t)-n_o(t).$
\begin{align}\nn
 n(t)  +   d\Phi^{sf}(t)/dt& \le n(t)   + c_1 \left(\frac{1}{\alpha}\right) n(t) + c_1\left(1-\frac{1}{\alpha}\right) n_o(t) + c_2 n_o(t) +c_1\left(-\max\left\{n_e(t)-n_o(t),0\right\}\right), \\  \nn
  & \stackrel{(a)}\le n(t)\left(1+ c_1\left(\frac{1}{\alpha}\right)-(1-\delta) c_1 \right) + n_o(t)\left(c_1\left(1+\left(1-\frac{1}{\alpha}\right)\right)+c_2\right),\\  \label{eq:runcond2}
  & \le n_o(t)\left(c_1\left(1+\left(1-\frac{1}{\alpha}\right)\right)+c_2\right),
\end{align}
where $(a)$ follows since $n_e(t) \ge (1-\delta) n(t)$, and the final inequality follows for
$ c_1\ge \alpha/\left(\alpha(1-\delta)-1\right)$.

When $n_o(t) =0$, then we get that \eqref{eq:runcond} holds with a smaller constant $\kappa$.

Combining, all the conditions, we get that the competitive ratio is at most $\left(\frac{1}{1-\delta}+c_1\left(1+\frac{\frac{1}{1-\delta}-1}{\alpha}\right)+c_2\right)$ where 
$c_1\ge \alpha/\left(\alpha(1-\delta)-1\right)$ and $\delta c_2 \ge 1+\frac{c_1}{\alpha}$. Thus, the tightest bound is

$$\left(\frac{1}{1-\delta}+c_1\left(1+\frac{\frac{1}{1-\delta}-1}{\alpha}\right)+c_2\right)$$ $$ \le \frac{1}{\alpha(1-\delta)-1}\left[\frac{\alpha(1-\delta)}{\delta} + \frac{\alpha(1-\delta)+\delta}{1-\delta}\right] .$$
\end{proof} 

%
%
%
%
\section{Numerical results}\label{sec:sim}
In this section, we present simulation results for the mean flow time (per job). We compare the performance of the proposed algorithm \textsc{Fractional-LCFS} with other known algorithms such as inelastic first \textsc{IF} \cite{berg2021case},  \textsc{EQUI} \cite{edmonds2000scheduling} and phase-aware FCFS \textsc{PA-FCFS} \cite{phaseawareFCFS}. 
With \textsc{PA-FCFS}, jobs are processed in the order in which they arrive, and the earliest arrived job is processed by as many servers as possible, i.e. if a job is in its inelastic phase then one server is allocated and other jobs are considered similarly over the remaining number of servers, while if a job is in its elastic phase then all the available servers are allocated to that.

For all simulations, we use $\alpha=2$.
In Fig. \ref{fig:basic}, we let the number of servers to be $N=10$, and consider a slotted time system, and plot the per-job flow time as a function of the per-slot mean arrival rate $\textsf{arr}$, where in each slot, 
the number of jobs arriving is Poisson distributed with the respective $\textsf{arr}$. 
For each job, the first/last phase is equally likely to be an elastic/in-elastic phase, and the number of phases of each job is Poisson distributed with mean $7$. The choice of $7$ is dictated by real-world datasets \cite{d1}.
For each phase, each job's size is exponentially distributed with mean $5$. For each iteration, we generate jobs for $1000$ slots, and count its flow time, and iterate over 1000 iterations. For \textsc{Fractional-LCFS}, we choose $\theta = \frac{1}{4}$. 
For all the results, we compare the performance of different algorithms for the same realization of random variables, and then average it out.

As we see from Fig. \ref{fig:basic}, the performance of \textsc{Fractional-LCFS} is similar to the inelastic first \textsc{IF} and the  \textsc{EQUI} \cite{edmonds2000scheduling} algorithm, however, the mean flow time of the  \textsc{PA-FCFS} is approximately $2$ or $3$ times larger than that of the other algorithms. With 
$\alpha=2$, the limitation of \textsc{PA-FCFS} is that whenever the earliest arrived job in its elastic phase, the speed dedicated to it is $N^{1/2}$ and no other job is processed. All the other three algorithms, in contrast, process multiple jobs with total speed roughly equal to $n (N/n)^{1/2}$ ($n$ is the number of outstanding jobs), thus having a far better performance.

Next, we repeat the simulations with setting of Fig. \ref{fig:basic} with increased number of servers $N=100$ in Fig. \ref{sim:100} to demonstrate the effect of load (the ratio of the mean per-slot job arrival rate and the number of servers) on the mean flow time. With $N=100$, the performance of \textsc{IF} is much improved since with large number of servers, the possibility that 
an in-elastic job blocks sufficiently many elastic jobs becomes smaller.  \textsc{Fractional-LCFS} with $\beta=1$ continues to outperform all other algorithms as shown in Fig. \ref{fig:basic}.

In Fig. \ref{fig:beta}, we plot the performance of \textsc{Fractional-LCFS} for different choices of $\beta$ with $\theta =1/4$ for mean per-slot arrival rate of $10$, and the rest of 
settings are the same as in Fig. \ref{fig:basic}. In the theoretical result we showed that for $\alpha=2$ with 
$\beta=1/6$ and $\theta = 1/72$, the competitive ratio of \textsc{Fractional-LCFS} is a constant. From Fig. \ref{fig:beta} we observe that in fact the performance of \textsc{Fractional-LCFS} improves by choosing larger values of $\beta$ and $\theta$, and the choice of $\beta=1/6$ and $\theta = 1/72$ was needed only for 
theoretical purposes. Fig. \ref{fig:beta} shows that $\beta=1$ has the best performance among 
different choices of $\beta$ for \textsc{Fractional-LCFS}.

Finally, in Fig. \ref{sim:worstcase}, we plot the mean flow time (per job) of the considered algorithms with $N=10$ for an arbitrary input where for each job the number of phases is $8$, and  the job size profile for each job is $[1, 10, 1, 10, 1, 10, 1, 10]$, where the first phase is elastic or in-elastic with equal probability. We see that in this case, the performance of \textsc{IF} deteriorates on account of highly heterogenous job sizes in the elastic and inelastic phases, while \textsc{EQUI} has similar performance to \textsc{Fractional-LCFS}.

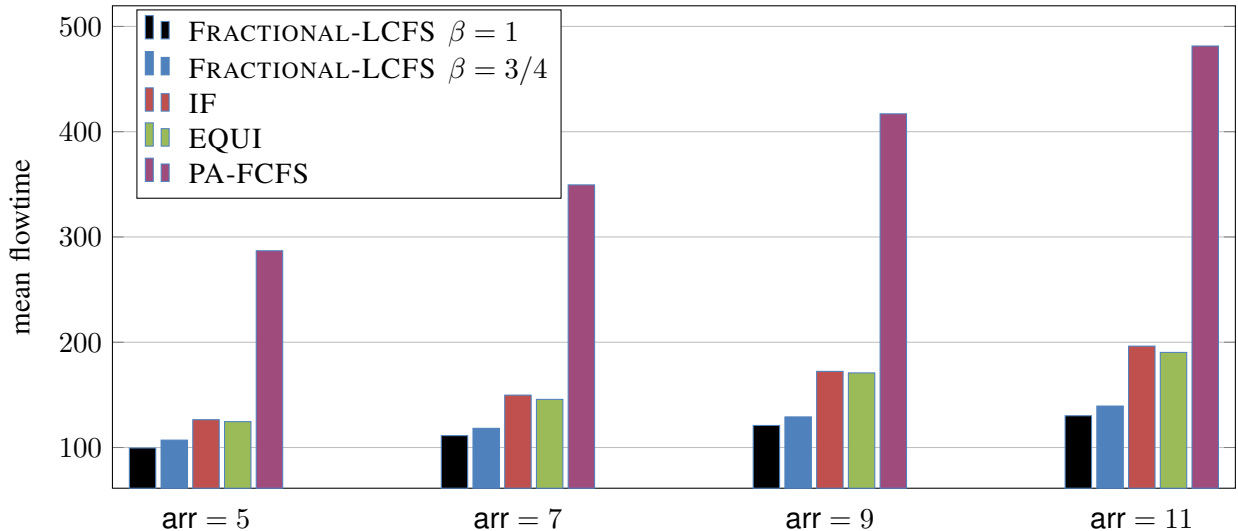
\begin{figure}
\centering
\begin{tikzpicture}
    \begin{axis}[
        width  = 1*\textwidth,
        height = 8cm,
        major x tick style = transparent,
        ybar,
        bar width=10pt,
        ymajorgrids = true,
        ylabel = {mean flowtime},
        symbolic x coords={$\textsf{arr}=5$, $\textsf{arr}=7$, $\textsf{arr}=9$, $\textsf{arr}=11$},
        xtick = data,
        scaled y ticks = false,
        legend cell align=left,
        legend style={
                at={(.4,.6)},
                anchor=south east,
                column sep=1ex}
    ]

     \addplot[style={bblue,fill=black,mark=none}]
            coordinates {($\textsf{arr}=5$, 99.5204) ($\textsf{arr}=7$,111.2413) ($\textsf{arr}=9$,120.9422) ($\textsf{arr}=11$, 130.1422)};

        \addplot[style={bblue,fill=bblue,mark=none}]
            coordinates {($\textsf{arr}=5$, 107.0254) ($\textsf{arr}=7$,118.2318) ($\textsf{arr}=9$,129.0677) ($\textsf{arr}=11$,139.3755)};

        \addplot[style={bblue,fill=rred,mark=none}]
            coordinates {($\textsf{arr}=5$, 126.4096) ($\textsf{arr}=7$,149.7868) ($\textsf{arr}=9$, 172.3358) ($\textsf{arr}=11$,196.3062)};

        \addplot[style={bblue,fill=ggreen,mark=none}]
            coordinates {($\textsf{arr}=5$, 124.5763) ($\textsf{arr}=7$,145.7353) ($\textsf{arr}=9$,170.8807) ($\textsf{arr}=11$,190.3415)};

        \addplot[style={bblue,fill=ppurple,mark=none}]
            coordinates {($\textsf{arr}=5$, 286.8909) ($\textsf{arr}=7$,349.4728) ($\textsf{arr}=9$,416.9715) ($\textsf{arr}=11$,  481.4032)};
%

        \legend{\textsc{Fractional-LCFS} $\beta=1$, \textsc{Fractional-LCFS} $\beta=3/4$, \textsc{IF}, \textsc{EQUI}, \textsc{PA-FCFS}}
    \end{axis}
\end{tikzpicture}
\caption{Comparison of mean flow time with different algorithms as a function of mean arrival rate per slot with $10$ servers.}
\label{fig:basic} 
\end{figure}

\begin{figure}
\centering
\begin{tikzpicture}
    \begin{axis}[
        width  = 1*\textwidth,
        height = 8cm,
        major x tick style = transparent,
        ybar,
        bar width=10pt,
        ymajorgrids = true,
        ylabel = {mean flowtime},
        symbolic x coords={$\textsf{arr}=5$, $\textsf{arr}=7$, $\textsf{arr}=9$, $\textsf{arr}=11$},
        xtick = data,
        scaled y ticks = false,
        legend cell align=left,
        legend style={
                at={(.4,.6)},
                anchor=south east,
                column sep=1ex}
    ]

     \addplot[style={bblue,fill=black,mark=none}]
            coordinates {($\textsf{arr}=5$, 37.1859) ($\textsf{arr}=7$,39.3274) ($\textsf{arr}=9$,41.4600) ($\textsf{arr}=11$,  43.5284)};

        \addplot[style={bblue,fill=bblue,mark=none}]
            coordinates {($\textsf{arr}=5$, 38.7956) ($\textsf{arr}=7$,40.8244) ($\textsf{arr}=9$,43.4138) ($\textsf{arr}=11$,45.7484)};

        \addplot[style={bblue,fill=rred,mark=none}]
            coordinates {($\textsf{arr}=5$, 37.5140) ($\textsf{arr}=7$,39.3714) ($\textsf{arr}=9$,41.6933) ($\textsf{arr}=11$, 43.9288)};

        \addplot[style={bblue,fill=ggreen,mark=none}]
            coordinates {($\textsf{arr}=5$, 38.4812) ($\textsf{arr}=7$,39.9371) ($\textsf{arr}=9$,41.4168) ($\textsf{arr}=11$,42.7503)};

        \addplot[style={bblue,fill=ppurple,mark=none}]
            coordinates {($\textsf{arr}=5$, 97.3739) ($\textsf{arr}=7$,117.0967) ($\textsf{arr}=9$,136.0010) ($\textsf{arr}=11$, 153.7302)};
%

        \legend{\textsc{Fractional-LCFS} $\beta=1$, \textsc{Fractional-LCFS} $\beta=3/4$, \textsc{IF}, \textsc{EQUI}, \textsc{PA-FCFS}}
    \end{axis}
\end{tikzpicture}
\caption{Comparison of mean flow time with different algorithms as a function of mean arrival rate per slot with $100$ servers.}
\label{sim:100} 
\end{figure}
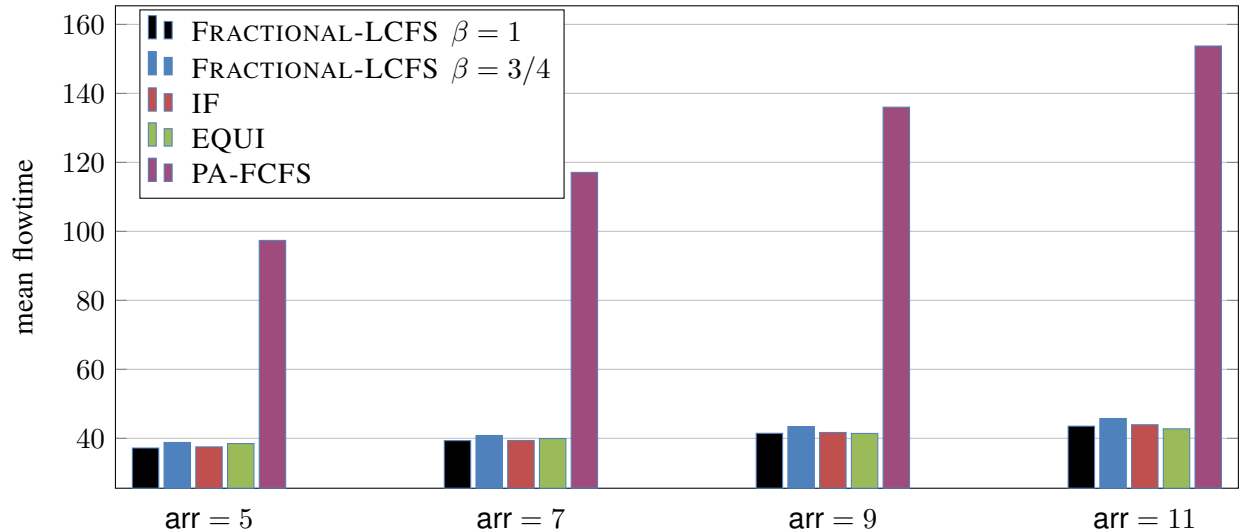

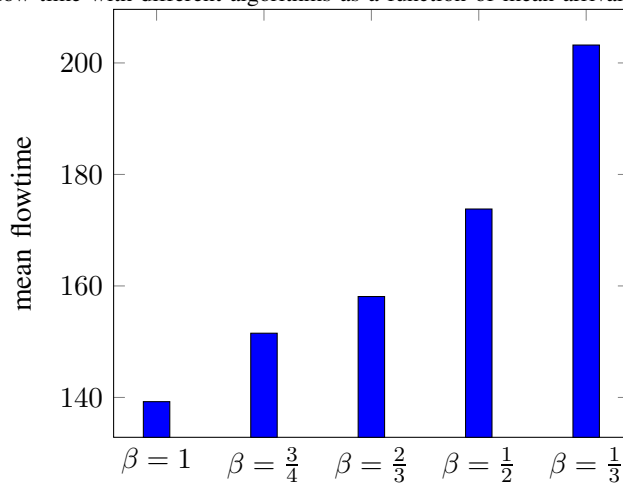
\begin{figure}
\centering
\begin{tikzpicture}
\begin{axis}[ylabel= mean flowtime,
    symbolic x coords={$\beta=1$, $\beta=\frac{3}{4}$, $\beta=\frac{2}{3}$, $\beta=\frac{1}{2}$, $\beta=\frac{1}{3}$},
    xtick=data]
    \addplot[ybar,fill=blue] coordinates {
        ($\beta=1$,139.2442)
        ($\beta=\frac{3}{4}$, 151.5209)
        ($\beta=\frac{2}{3}$,158.0990)
        ($\beta=\frac{1}{2}$,173.7982)
         ($\beta=\frac{1}{3}$,203.2104)
    };
\end{axis}
\end{tikzpicture}
\caption{Comparison of flow time for \textsc{Fractional-LCFS} with different choices of $\beta$ for $\theta = 1/4$ with mean per-slot arrival rate of $10$ with $10$ servers.}
\label{fig:beta}
\end{figure}

\begin{figure}
\centering
\begin{tikzpicture}
    \begin{axis}[
        width  = 1*\textwidth,
        height = 8cm,
        major x tick style = transparent,
        ybar,
        bar width=10pt,
        ymajorgrids = true,
        ylabel = {mean flowtime},
        symbolic x coords={$\textsf{arr}=5$, $\textsf{arr}=7$, $\textsf{arr}=9$, $\textsf{arr}=11$},
        xtick = data,
        scaled y ticks = false,
        legend cell align=left,
        legend style={
                at={(.4,.6)},
                anchor=south east,
                column sep=1ex}
    ]

     \addplot[style={bblue,fill=black,mark=none}]
            coordinates {($\textsf{arr}=5$, 91.1526) ($\textsf{arr}=7$,100.5989) ($\textsf{arr}=9$,110.1661) ($\textsf{arr}=11$,  119.3765)};

        \addplot[style={bblue,fill=bblue,mark=none}]
            coordinates {($\textsf{arr}=5$, 92.5731) ($\textsf{arr}=7$,102.5448) ($\textsf{arr}=9$,112.0136) ($\textsf{arr}=11$,120.2527)};

        \addplot[style={bblue,fill=rred,mark=none}]
            coordinates {($\textsf{arr}=5$, 120.4161) ($\textsf{arr}=7$,143.3853) ($\textsf{arr}=9$,169.8926) ($\textsf{arr}=11$, 197.4343)};

        \addplot[style={bblue,fill=ggreen,mark=none}]
            coordinates {($\textsf{arr}=5$, 98.5331) ($\textsf{arr}=7$,108.5039) ($\textsf{arr}=9$,117.3629) ($\textsf{arr}=11$,136.5428)};

        \addplot[style={bblue,fill=ppurple,mark=none}]
            coordinates {($\textsf{arr}=5$, 229.0917) ($\textsf{arr}=7$,275.4928) ($\textsf{arr}=9$,329.8866) ($\textsf{arr}=11$, 381.6000)};
%

        \legend{\textsc{Fractional-LCFS} $\beta=1$, \textsc{Fractional-LCFS} $\beta=3/4$, \textsc{IF}, \textsc{EQUI}, \textsc{PA-FCFS}}
    \end{axis}
\end{tikzpicture}
\caption{Comparison of mean flow time with different algorithms as a function of mean arrival rate per slot with $10$ servers for the arbitrary job sizes.}
\label{sim:worstcase} 
\end{figure}
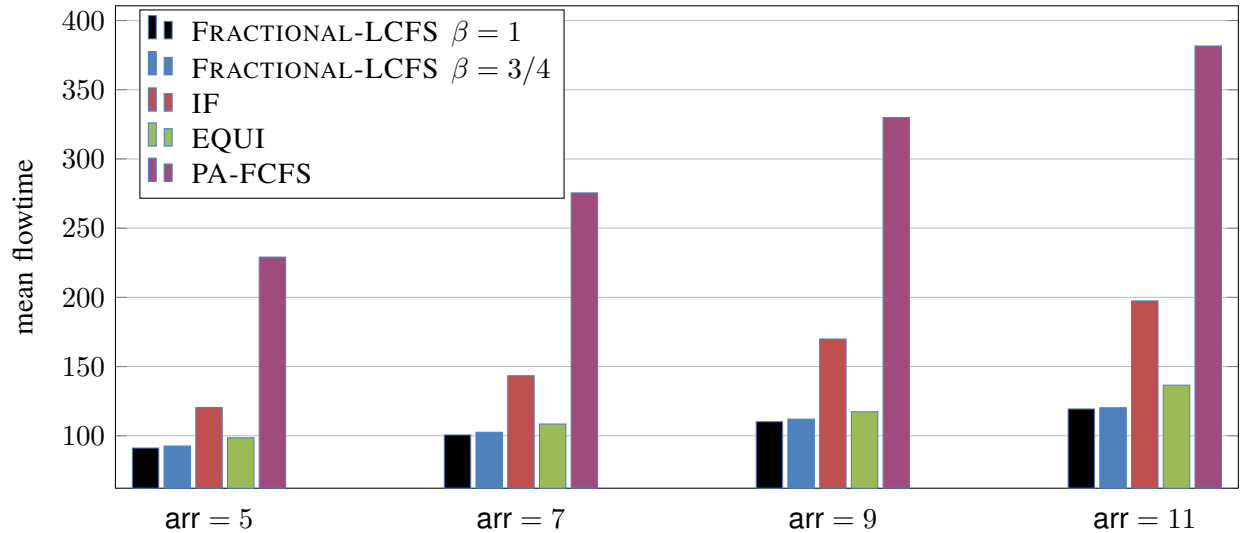

%
%

\section{Conclusions}
In this paper, we considered an important problem of flow time minimization in  data centers, where jobs migrate between two phases of parallelizability (called  elastic and in-elastic) multiple times. In the 
elastic phase, there is flexibility of parallelizing the job over multiple servers, while in the in-elastic phase, the job has to be processed by a single server. 
Moreover, in the elastic phase there is  limited parallelizability, and the speed increment diminishes as more and more servers are allocated to any job.
We considered the online setting, where jobs arrive over time with arbitrary sizes and arrival times, and proposed a LCFS type algorithm for scheduling,  that 
processes the  scheduled jobs with equal speed. We showed that  its competitive ratio is a constant that only depends on the speed-up exponent $\alpha$ as long as $\alpha>1$. 
In recent work, this model has been considered, however, for a specific stochastic input, where the size of the job in both the elastic and the in-elastic phase was exponentially distributed with identical parameters in the two phases, and 
a job always departed on completion of some in-elastic phase.
With the specific stochastic input, always scheduling as many  jobs that are in their in-elastic phases was shown to be optimal. 
With arbitrary input, our result overcomes fundamental difficulty found in literature where similar results were shown only in the presence of resource augmentation, by exploiting the 
specific structure of the problem with just two phases that is practically well motivated. We also considered the 
case when all jobs are available at time $0$, and for which a different algorithm has significantly 
better competitive ratio than the online jobs arrival case.
\bibliographystyle{IEEEtran}
\bibliography{refs}
  \appendix
  \section{}\label{app0}
  \subsection{Proof of Theorem \ref{thm:flowtimeplusenergyonline}} 

From here on we refer to algorithm \textsc{Fractional-LCFS} as just algorithm.
Let at time $t$, the set of outstanding (unfinished) number of jobs with the algorithm be $A(t)$ with $n(t) = |A(t)|$. Similarly, let $O(t)$ be the set of outstanding jobs with the $\opt$ at time $t$.
Let at time $t$, the {\bf rank} $r_j(t)$ of a job $j \in A(t)$  be equal to the number of outstanding jobs of $A(t)$ with the algorithm that have arrived before job $j$. Note that the rank of a job does not change on arrival of a new job, but can change if a job departs that had arrived earlier.

Let $Q(x) = \frac{x}{P(x)}$. which specializes to  $Q(x) = x^{1-\frac{1}{\alpha}}$ for $P(x) = x^{1/\alpha}$. Moreover, let $\{x\}^+ = \max\{x,0\}$.
Then we consider the following potential function 
\begin{equation}\label{defn:phi}
\Phi(t) = c_1 \Phi_1(t) + c_2\Phi_2(t), \end{equation}
where \begin{equation}\label{phi1}
\Phi_1(t) = \sum_{j\in A(t)} \frac{r_j(t)}{P(N)Q(r_j(t))} \left(w_j^A(t) - w_j^o(t)\right)^+,
\end{equation}
and
\begin{equation}\label{phi2}
\Phi_2(t)= \sum_{j\in A(t)} {\bar w}_{j\iota}^A(t) - \sum_{j\in O(t)} {\bar w}_{j\iota}^o(t),
\end{equation}
where $w_j^A(t)$ ($w_j^o(t)$) is the remaining size (sum of the job sizes of all the remaining elastic and in-elastic phases) of job $j$ with the algorithm ($\opt$) at time $t$, while ${\bar w}_{j\iota}^A(t)$ (${\bar w}_{j\iota}^o(t)$) is the sum of the remaining size of job $j$ in all its remaining in-elastic phases  with the algorithm ($\opt$) at time $t$, and $c_1,c_2$ are constants to be chosen later.  

We next show that the potential function $\Phi(t)$ satisfies the second boundary condition. The fact that the first boundary condition is satisfied is trivial.
\begin{lemma}\label{lem:jumplcfs}
Potential function $\Phi(t)$ \eqref{defn:phi} does not change on arrival of any new job. Moreover, on a departure of a job with the algorithm or the $\opt$, the potential function $\Phi(t)$ \eqref{defn:phi} does not increase.
\end{lemma}

The proofs of Lemma \ref{lem:jumplcfs} and \ref{lem:optphi} are provided in Appendix \ref{app1}.

We next bound the drift $d\Phi(t)/dt$ because of the processing by the $\opt$,  and the algorithm, respectively. To avoid cumbersome notation, we write $\beta n(t)$  or $\theta n(t)$ instead of $\lceil \beta n(t)\rceil$ or $\lceil \theta n(t)\rceil$everywhere.

%

\begin{lemma}\label{lem:optphi}
The change in the potential function \eqref{defn:phi} because of the $\opt's$ contribution
 \begin{align} 
 d\Phi(t)/dt & \le 
 c_1 n(t) \frac{Q(n_o(t))}{Q(n(t))}  + c_2n_o(t).\end{align}
\end{lemma}

 \begin{lemma}\label{lem:algphi}
 With $0 < \theta +  \gamma < \beta$, for any $t$ where $n_o(t) \le \gamma n(t)$, the change in the potential function \eqref{defn:phi} because of the algorithm's contribution is $ d\Phi(t)/dt $
 \begin{align} 
& \le \begin{cases} - c_1 \frac{(1-\beta) (\beta - \gamma) n(t) }{P(\beta)} &  \text{if} \ \frac{N}{\beta n(t)} \le 1, \\
  -c_2 \min\{N, n_\iota(t)\} &  \text{if} \ \frac{N}{\beta n(t)} > 1 \ \text{and} \ n_\iota(t) \ge \theta n(t),\\
  - c_1 \frac{(1-\beta) (\beta -\theta- \gamma) n(t) }{P(\beta)},& \text{otherwise}.
   \end{cases}
 \end{align}
\end{lemma}
The proof of Lemma \ref{lem:algphi} is provided in Appendix \ref{app2}.
  To prove Theorem \ref{thm:flowtimeplusenergyonline}, we check the running condition \eqref{eq:runcond} for the following two cases separately for a fixed $\gamma$ such that $\theta+\gamma < \beta$ (choice to be made later) : i) $n_o(t) > \gamma n(t)$ and  ii) $n_o(t) \le \gamma n(t)$, and show that it holds for a constant $\kappa$.

Case i) $n_o(t) > \gamma n(t)$. In this case, we only count the $\opt's$ contribution to $d\Phi(t)/dt$, which is sufficient since the algorithm's contribution to $d\Phi(t)/dt$ is always non-positive. 
From Lemma \ref{lem:optphi}, we have that 
\begin{align}\nn
  n(t) +  d\Phi(t)/dt & \le n(t)  + c_1 n(t) \frac{Q(n_o(t))}{Q(n(t))} +c_2n_o(t),   \\ \nn
  &  \stackrel{(a)}\le n(t) + c_1 n(t) \frac{Q(b n(t))}{Q(n(t))} +c_2n_o(t) , \\ \nn
  &  =  n(t) + c_1 n(t) b^{1-1/\alpha} + c_2n_o(t), \\  \nn
  &   \stackrel{(b)}\le  n(t)  + c_1 n(t) b +c_2n_o(t), \\ 
  & =  n(t)  + c_1 n_o(t) +c_2n_o(t), \\  \label{eq:optcont1}
    &  \stackrel{(c)}\le  (1/\gamma+c_1+c_2 ) n_o(t),
\end{align}
where in $(a)$ we let $n_o(t) = b n(t)$ and  inequality $(b)$ follows when $b>1$. 
Finally $(c)$ follows since $n_o(t) > \gamma n(t)$.
When $b<1$, then $\frac{Q(b n(t))}{Q(n(t))} < 1$. Thus, similar to \eqref{eq:optcont1}, for $b<1$, we get 
\begin{align}\nn
  n(t) + d\Phi(t)/dt & \le n(t) + c_1 n(t) \frac{Q(n_o(t))}{Q(n(t))}+c_2n_o(t),   \\ \nn
    & \le  n(t)(1+c_1)+c_2n_o(t), \\ \label{eq:optcont2}
    & \le \left(\frac{1+c_1}{\gamma}+c_2\right) n_o(t).
\end{align}
%

Case ii) $n_o(t) \le \gamma n(t)$. 
Let $n_o(t)>0$. 

ii-a) With $\frac{N}{\beta n(t)} \le 1$, from Lemma \ref{lem:optphi} and Lemma \ref{lem:algphi}, \eqref{eq:runcond} can be bounded as  $n(t) + d\Phi(t)/dt  $
\begin{align}\nn
 & \le n(t) +  c_1 n(t) \frac{Q(n_o(t))}{Q(n(t))} + c_2n_o(t)  -c_1 \frac{(1-\beta) (\beta - \gamma)}{P(\beta)} n(t) ,\\ \nn
  &  \stackrel{(a)}\le  c_2n_o(t) + n(t) \left(1+ c_1\left(\gamma^{1-1/\alpha} - \frac{(1-\beta) (\beta - \gamma)}{P(\beta)}\right)\right), \\ \label{eq:rel1}
  & \stackrel{(b)} \le   c_2n_o(t),
  \end{align}
  where $(a)$ follows since $n_o(t)\le \gamma n(t)$, while $(b)$ follows
 for choice of $\gamma, \beta, c$ that satisfy
   \begin{equation}\label{eq:cond1}
\frac{(1-\beta) (\beta - \gamma)}{P(\beta)} > \gamma^{1-1/\alpha} \ \text{and} \  c_1 \ge  \frac{-1}{\left(\gamma^{1-1/\alpha} - \frac{(1-\beta) (\beta - \gamma)}{P(\beta)} \right)}.
\end{equation}
ii-b) When $\frac{N}{\beta n(t)} > 1$ and $n_\iota(t) \ge \theta n(t)$, from Lemma \ref{lem:optphi} and Lemma \ref{lem:algphi}, \eqref{eq:runcond} can be bounded as   $n(t) + d\Phi(t)/dt  $
\begin{align}\nn
 & \le n(t) +  c_1 n(t) \frac{Q(n_o(t))}{Q(n(t))} + c_2n_o(t)-  c_2 \min\{N, n_\iota(t)\},\\\nn
  &  \stackrel{(a)}\le  c_2n_o(t) + n(t) \left(1+ c_1\gamma^{1-1/\alpha} - c_2\theta\right), \\\label{eq:rel2}
  & \stackrel{(b)} \le   c_2n_o(t),
  \end{align}
  where $(a)$ follows since $n_o(t)\le \gamma n(t)$, $n_\iota(t) \ge \theta n(t), \frac{N}{\beta n(t)}>1$ and $\theta < \beta$, while $(b)$ follows
 for 
 \begin{equation}\label{eq:cond2}
c_2 \ge \frac{(1+ c_1\gamma^{1-1/\alpha})}{\theta}.
\end{equation}

ii-c) Finally, when $\frac{N}{\beta n(t)} > 1$ and $n_\iota(t) < \theta n(t)$, from Lemma \ref{lem:optphi} and Lemma \ref{lem:algphi}, \eqref{eq:runcond} can be bounded as $n(t) + d\Phi(t)/dt  $ 
\begin{align}\nn
 & \le n(t) +   \frac{c_1 n(t)Q(n_o(t))}{Q(n(t))} -  \frac{c_1(1-\beta) (\beta - \theta- \gamma) n(t)}{P(\beta)} \\ \nn
 &\quad \quad + c_2n_o(t),\\ \nn
  &  \stackrel{(a)}\le  c_2n_o(t) + n(t) \left(1+ c_1\left(\gamma^{1-1/\alpha} - \frac{(1-\beta) (\beta -\theta- \gamma)}{P(\beta)}\right)\right), \\ \label{eq:rel3}
  & \stackrel{(b)} \le   c_2n_o(t),
  \end{align}
  where $(a)$ follows since $n_o(t)\le \gamma n(t)$, while $(b)$ follows
 for choice of $\gamma, \beta, c$ that satisfy
   \begin{equation}\label{eq:cond3a}
\frac{(1-\beta) (\beta - \theta- \gamma)}{P(\beta)} > \gamma^{1-1/\alpha}\end{equation} and 
\begin{equation}\label{eq:cond3b}   c_1 \ge  \frac{-1}{\left(\gamma^{1-1/\alpha} - \frac{(1-\beta) (\beta - \theta- \gamma)}{P(\beta)} \right)}.
\end{equation}

When $n_o(t)=0$, the $\opt$'s contribution to $d\Phi(t)/dt$ is zero, and we can bound  \eqref{eq:runcond} with smaller value of $\kappa$.
Combining \eqref{eq:rel1}, \eqref{eq:rel2}, \eqref{eq:rel3},  together with \eqref{eq:optcont1} and \eqref{eq:optcont2},
the competitive ratio of 
the proposed algorithm  is  at most
\begin{equation}\label{eq:finalbound}
\frac{1+c_1}{\gamma}+c_2
\end{equation}
 for $\beta, \theta, \gamma$, that satisfy \eqref{eq:cond1}, \eqref{eq:cond2}, \eqref{eq:cond3a} and \eqref{eq:cond3b}.
Depending on $\alpha >1$, there exists a $\beta < 1$ satisfying \eqref{eq:cond1}, \eqref{eq:cond3a}  and \eqref{eq:cond3b} with $\theta = \gamma = \beta^2/2$ as follows. In particular, with 
 $\theta = \gamma = \beta^2/2$, to satisfy \eqref{eq:cond1}, \eqref{eq:cond3a} and \eqref{eq:cond3b} i.e., $\frac{(1-\beta) (\beta - \gamma)}{P(\beta)} > \gamma^{1-1/\alpha}$ and $ \frac{(1-\beta) (\beta - \theta- \gamma)}{P(\beta)} > \gamma^{1-1/\alpha}$, it is sufficient that $1-2\beta +\beta^2 > \frac{\beta}{2}^{1-1/\alpha}$. Since $\alpha >1$, $1-2\beta +\beta^2 -\left(\frac{\beta}{2}\right)^{1-1/\alpha} = 1$ at $\beta=0$. Thus, using continuity, we know that there exists a $0 < \beta < 1$ satisfying \eqref{eq:cond1}, \eqref{eq:cond3a} and \eqref{eq:cond3b} with $\theta = \gamma = \beta^2/2$.
This implies that the competitive ratio is a constant that only depends on $\alpha$ and not on any 
other system parameter. Moreover, notice that as $\alpha \rightarrow 1$, the appropriate choice of $\beta$ decreases implying that the competitive ratio \eqref{eq:finalbound} increases. 

For example, for $\alpha=2$, let $\beta = \frac{1}{6}$ and $\theta= \gamma = \beta^2/2$,  $c_1= \frac{-1}{\left(\gamma^{1-1/\alpha} - \frac{(1-\beta) (\beta - \theta- \gamma)}{P(\beta)} \right)}=6.06$, $c_2= \frac{(1+c_1\gamma^{1-1/\alpha})}{\theta} = 72(1+.77)=127.44$. We get a competitive ratio of 
$\frac{1+c_1}{\gamma}+c_2 \le 72\times (1+6.06)  + 127.44= 635.76$.

Analytically optimizing the competitive ratio with respect to the variables, $\beta, \alpha$, and $\gamma$  could result in a much lower bound, however, appears difficult. Numerically, however, one can easily do so.

\begin{remark}\label{rem:choiceofbeta} For any $\alpha>1$, choosing $\theta = \gamma = \beta^2/2$, and $0<\beta<1$ such that $1-2\beta +\beta^2 > \frac{\beta}{2}^{1-1/\alpha}$ is sufficient to make the competitive ratio constant. Moreover, finding such a $\beta$ is easy numerically.
\end{remark}
\section{}\label{app1}
\begin{proof}[Proof of Lemma \ref{lem:jumplcfs}]
 On an arrival of a new job $j$, the ranks of all the existing jobs do not change, while for the newly arrived job $j$, $w_j^A(t) - w_j^o(t)=0$.
  Hence the potential function $\Phi_1(t)$ \eqref{defn:phi} does not change on arrival of any new job. 
  
 On a departure of a job with the algorithm, rank of any remaining job can only decrease, in particular by $1$. 
 Thus, if at time $t$ when job $k$ departs with the algorithm, job $j$'s ($j\in A(t^{+})$) rank at time  $t^{+}$, is either $r_j(t^+) =  r_j(t)$ or 
 $r_j(t^+) =  r_j(t)-1$. Since function $\frac{r_j(t)}{Q(r_j(t))}$ is a non-decreasing function, thus the potential function $\Phi_1(t)$ does not increase on departure of a job with the algorithm. 
 
%
%
For the $\opt$, only $w_j^o(t)$ decreases with job processing and that too smoothly. Thus,  there is no discontinuity when  a job departs with the $\opt$, hence $\Phi_1(t)$ does not change when   a job departs with the $\opt$.

Moreover, for $\Phi_2(t)$, on an arrival of a new job ${\bar w}_{j\iota}^A(t) - {\bar w}_{j\iota}^o(t)=0$, while 
 there is no discontinuity when a job departs with the $\opt$ or the algorithm, since both ${\bar w}_{j\iota}^A(t)$ and ${\bar w}_{j\iota}^o(t)$ decrease smoothly.  Hence $\Phi_2(t)$ does not change when  a new job arrives or a job departs with the $\opt$ or the algorithm.

\end{proof}
\begin{proof}[Proof of Lemma \ref{lem:optphi}]

We begin with the following simple result whose proof is immediate.
\begin{lemma}\label{lem:optmaxspeedspecial}
Disregarding the unit speed constraint for any job whose in-elastic part is being processed,  the maximum speed devoted to processing any one job by the $\opt$ is at most $P(N)$. 
Moreover, the sum of the speeds with which $\opt$ is processing any of its $k$ jobs is at most $Q(k)P(N)$.
\end{lemma}

From the definition of $\Phi(t)$ \eqref{defn:phi},  $\opt$ can increase $\Phi_1(t)$ at time $t$ only if it processes jobs that also belong to the set $A(t)$ (outstanding jobs with the algorithm). Thus, from Lemma \ref{lem:optmaxspeedspecial}, the maximum sum of the speeds  devoted to the  set of $A(t)$ jobs by the $\opt$ is at most  
$Q(n(t)) P(N)$, where each job gets processed at speed $P\left(\frac{N}{n(t)}\right)$. Moreover, since $\opt$ contains only $n_o(t)$ jobs, sum of the speeds devoted to the 
$n(t)$ jobs of the algorithm is at most $$Q(\min\{n(t), n_o(t)\}) P(N).$$ 
From the definition 
of  $\Phi_1(t)$ \eqref{phi1}, the maximum increase in $\Phi_1(t)$ is possible if the total speed of the $\opt$ that it can dedicate to jobs belonging to $A(t)$ is dedicated to the single job with the largest rank among $A(t)$, i.e., the job with rank equal to $n(t)$. Thus, because of processing by the $\opt$\begin{align}\nn
  d\Phi_1(t)/dt  &\le   \frac{n(t)}{P(N) Q(n(t))}  \times Q(\min\{n(t), n_o(t)\}) P(N), \\
   & \le  n(t) \frac{Q(n_o(t))}{Q(n(t))}.
\end{align}

Moreover, any job that is in its in-elastic phase  can be processed with at most unit speed. Since there are at most $n_o(t)$ jobs with the $\opt$ that 
are in their in-elastic phases, we get  
$$d\Phi_2/dt \le n_o(t).$$

   \end{proof}

\section{}\label{app2}
\begin{proof}[Proof of Lemma \ref{lem:algphi}]
  
  Case I : $\left(\frac{N}{\beta n(t)}\right) \le 1$
   Since the algorithm executes the $\beta n(t)$ jobs that have arrived most recently, the rank of job $i$ that is being processed by the algorithm is 
  $r_i(t) = n(t) - i+1$ for $i=1, \dots, \beta n(t)$. Since $n_o(t) \le \gamma n(t)$, and $\gamma < \beta$,  $$w_j^A(t) - w_j^o(t) > 0$$ for at  least 
  $(\beta - \gamma) n(t)$ jobs with the  algorithm. In the worst case, the ranks of these $(\beta - \gamma) n(t)$ jobs are 
  $(1-\beta) n(t) + i-1$ for $i=1, \dots, (\beta-\gamma) n(t)$.
  
  Since the speed for any of the job executed by the algorithm is $s(t) = P\left(\frac{N}{\beta n(t)}\right)$, the change in the potential function because of the algorithm's  
  processing to $\Phi_1(t)$ is 
  \begin{align*}\label{}
d\Phi_1(t)/dt\le  &  -\sum_{i= (1-\beta) n(t)}^{(1-\beta) n(t) + (\beta - \gamma) n(t)} \frac{r_i(t)}{P(N)Q(r_i(t))}  P\left(\frac{N}{\beta n(t)}\right),\\
  \stackrel{(a)} \le &   - \sum_{i= (1-\beta) n(t)}^{(1-\beta) n(t) + (\beta - \gamma) n(t)} \frac{r_i(t)}{Q(n(t))}  \frac{1}{ P(\beta n(t))}, \\
  = &   - \sum_{i= (1-\beta) n(t)}^{(1-\beta) n(t) + (\beta - \gamma) n(t)} \frac{r_i(t)}{Q(n(t))}  \frac{1}{ P(n(t))} \frac{1}{P(\beta)} , \\
  \stackrel{(b)} = &   - \sum_{i= (1-\beta) n(t)}^{(1-\beta) n(t) + (\beta - \gamma) n(t)}  \frac{r_i(t)}{n(t)} \frac{1}{P(\beta)}, \\
 \stackrel{(c)}\le &  -  \frac{(1-\beta) (\beta - \gamma) n(t) n(t)}{\beta } \frac{1}{n(t) P(\beta)} , \\
  = & -  \frac{(1-\beta) (\beta - \gamma) n(t) }{P(\beta)} , 
\end{align*}
where $(a)$ follows since $r_i(t) \le n(t)$ and $$\frac{P\left(\frac{N}{\beta n(t)}\right)}{P(N)} \ge \frac{1}{ P(\beta n(t))},$$ while $(b)$  follows since $Q(x) P(x) = x$, and finally $(c)$ follows since there are $(\beta-\gamma) n(t)$ jobs that are being executed each with rank at least $(1-\beta) n(t)$.
For $\Phi_2(t)$, in this case, we just bound $d\Phi_2(t)/dt \le 0$ because of the algorithm's processing.

 Case II : $\left(\frac{N}{\beta n(t)}\right) > 1$

 IIa:  $n_\iota(t) \ge \theta n(t)$
 In this case, for the algorithm we will only consider the drift $d\Phi_2(t)/dt$, and trivially upper bound $d\Phi_1(t)/dt \le 0$. When $n_\iota(t) \ge \theta n(t)$, each of the $\min\{N, n_\iota(t)\}$ jobs are processed 
 at unit speed by the algorithm, and we get 
 \begin{equation}\label{}
d\Phi_2(t)/dt \le - \min\{N, n_\iota(t)\}.
\end{equation}
 IIb: $n_\iota(t) < \theta n(t)$
 In this case, for the algorithm we will only consider the drift $d\Phi_1(t)/dt$ and upper bound $d\Phi_2(t)/dt \le 0$.
 
 In this case, the algorithm executes those jobs that are in their elastic phases among the $\beta n(t)$ jobs that have arrived most recently.
 Since $n_\iota(t) < \theta n(t)$, and $n_\iota(t)+ n_e(t) = n(t)$, at least $(\beta-\theta)n(t)$ jobs (that are in their elastic phases) are being processed.
 
Moreover, since $n_o(t) \le \gamma n(t)$,  for at  least 
  $(\beta - \theta- \gamma) n(t)$ jobs that are being processed by the algorithm $$w_j^A(t) - w_j^o(t) > 0,$$ and the rank of each of these $(\beta - \theta -\gamma) n(t)$ jobs is at least  
  $(1-\beta) n(t)$.
  
  Since the speed for any of the job executed by the algorithm is $s(t) = P\left(\frac{N}{\beta n(t)}\right)$, the change in the potential function $\Phi_1(t)$ because of the algorithm's  
  processing is 
  \begin{align*}\label{}
d\Phi_1(t)/dt\le  &  - \sum_{i=1}^{(\beta-\theta-\gamma)n(t)} \frac{r_i(t)}{P(N)Q(r_i(t))}  P\left(\frac{N}{\beta n(t)}\right),\\
  \stackrel{(a)} \le &   - \sum_{i=1}^{(\beta-\theta-\gamma)n(t)} \frac{r_i(t)}{Q(n(t))}  \frac{1}{ P(\beta n(t))}, \\
  = &   - \sum_{i=1}^{(\beta-\theta-\gamma)n(t)} \frac{r_i(t)}{Q(n(t))}  \frac{1}{ P(n(t))} \frac{1}{P(\beta)} , \\
  \stackrel{(b)} = &   - \sum_{i=1}^{(\beta-\theta-\gamma)n(t)}  \frac{r_i(t)}{n(t)} \frac{1}{P(\beta)}, \\
 \stackrel{(c)}\le &  - \frac{(1-\beta) (\beta - \theta - \gamma) n(t) n(t)}{\beta } \frac{1}{n(t) P(\beta)} , \\
  = & -  \frac{(1-\beta) (\beta - \theta- \gamma) n(t) }{P(\beta)} , 
\end{align*}
where $(a)$ follows since $r_i(t) \le n(t)$, while $(b)$  follows since $Q(x) P(x) = x$, and finally $(c)$ follows since there are $(\beta-\theta- \gamma) n(t)$ jobs that are being executed each with rank at least $(1-\beta) n(t)$.

\end{proof}

\section{}\label{app3}
\begin{proof}[Proof of Lemma \ref{lem:jumpequi}]
First we argue about $\Phi^{sf}_1(t)$.
 On a departure of a job with the algorithm or the $\opt$, $n^i(t,q)$ or $n_o^i(t,q)$ changes for only $q=0$, and since there is an integral outside,  $\int_{0}^\infty ( n^i(t,q) - n_o^i(t,q))^+ dq$ remains the same on a departure of a job with either the algorithm or the $\opt$. 

The pre-factor term $P\left(\frac{n(t)}{N}\right)$ changes though, however only decreases, when  there is a departure of a job with the algorithm, on account of $n(t) \rightarrow n(t)-1$. Since the integral is always non-negative, overall, the potential function can only decrease if at on account of a departure with the algorithm.
 Moreover, the departure of any job with the $\opt$ does not change the pre-factor.
 Since the integral is always non-negative, thus $\Phi^{sf}_1(t)$ does not increase on a departure with the algorithm or the $\opt$.
 
 For the $\Phi^{sf}_2(t)$, there are no discontinuities, thus $\Phi^{sf}_2(t)$ also does not increase on a departure with the algorithm or the $\opt$.
\end{proof}
   \section{}\label{app4}
   \begin{proof}[Proof of Lemma \ref{lem:driftphisopt}]
   From Lemma \ref{lem:optmaxspeedspecial}, we know that the sum of the speeds used by the $\opt$ over its $n_o(t)$ jobs is at most 
\begin{equation}\label{eq:dummy400}
\sum_{i=1}^{n_o(t)} s^o_i(t) \le Q(n_o(t)) P(N).
\end{equation}

Using this, we bound the drift $d \Phi^{sf}_1(t)/dt$ with respect to processing by the $\opt$ as follows
\begin{align} \nn
d \Phi^{sf}_1(t)/dt &\stackrel{(a)} \le P\left(\frac{n(t)}{N}\right)\left(\sum_{i=1}^{n_o(t)} s^o_i(t)\right),\\    \nn
  &\stackrel{(b)} \le P\left(\frac{n(t)}{N}\right)Q(n_o(t)) P(N),   \\ \nn
  &\le    P(n(t))Q(n_o(t)), \\   \nn
  & =  n(t)^{1/\alpha} n_o(t) ^{1-1/\alpha}, \\    \label{eq:phiboundsf2}
  & \stackrel{(c)}\le   \left(\frac{1}{\alpha}\right) n(t) + \left(1-\frac{1}{\alpha}\right) n_o(t). 
\end{align}
where for $(a)$ we assume that all the $n_o$ jobs of the $\opt$ are getting processed at non-zero speed (best case in terms of increasing $d \Phi^{sf}(t)/dt$), while $(b)$ follows from \eqref{eq:dummy400}, and 
$(c)$ 
follows from the generalized AM-GM inequality. \footnote{For $a_i\ge 0$ and $\lambda_i\ge 0$ with $\sum_{i=1}^n \lambda_i=1$, then 
$\prod_{i=1}^n c_i^{\lambda_i} \le \sum_{i=1}^n \lambda_i c_i$.} 

Moreover, any job that is in its in-elastic phase  can be processed with at most unit speed. Since there are at most $n_o(t)$ jobs with the $\opt$ that 
are in their in-elastic phases, we get  
$$d\Phi^{sf}_2/dt \le n_o(t).$$
\end{proof}
\section{}\label{app5}
\begin{proof}[Proof of Lemma \ref{lem:driftphisalg}]
Case I $n(t) \ge N$. In this case, note that for at least $\max\{n(t)-n_o(t),0\}$ jobs belonging to $A(t)$, the corresponding terms $( n^i(t,q) - n_o^i(t,q))^+ > 0$ in $\Phi_1^{sf}(t)$. Thus, algorithm \textsc{PA-EQUI} is decreasing work at speed $s_i(t)$ for at least $\max\{n(t)-n_o(t),0\}$ jobs. Hence, the drift $d \Phi^{sf}_1(t)/dt$ with respect to processing by the algorithm \textsc{PA-EQUI} is
\begin{align}\nn
d \Phi^{sf}_1(t)/dt & \stackrel{(a)}= P\left(\frac{n(t)}{N}\right)\left(-\max\{n(t)-n_o(t),0\} s_i(t)  \right),\\   \label{eq:phiboundsf1}
  &\stackrel{(b)} =  \left(-\max\{n(t)-n_o(t),0\}  \right)
\end{align}
where $(a)$ follows since for at least $\max\{n(t)-n_o(t),0\}$ jobs, the \textsc{PA-EQUI} algorithm is decreasing work at speed $s_i(t)$, while $(b)$ follows since 
$s_i(t)=P\left(\frac{N}{n(t)}\right)$ for all jobs $i$ being processed by the \textsc{PA-EQUI} algorithm. 

Case II  $n(t) < N$ and $n_\iota(t) \ge \delta n(t)$. In this case, for the \textsc{PA-EQUI} algorithm we will only consider the drift $d\Phi_2(t)/dt$, and trivially upper bound $d\Phi^{sf}_1(t)/dt \le 0$. When $n_\iota(t) \ge \delta n(t)$, each of the $n_\iota(t)$ jobs are processed 
 at unit speed by the algorithm, and we get 
 \begin{equation}\label{eq:phiboundsf2}
d\Phi_2(t)/dt \le - n_\iota(t).
\end{equation}
Case III  $n(t) < N$ and $n_\iota(t) < \delta n(t)$ or equivalently  $n_e(t) \ge (1-\delta) n(t)$.
In this case, note that for at least $\max\{n_e(t)-n_o(t),0\}$ jobs belonging to $A(t)$, the corresponding terms $( n^i(t,q) - n_o^i(t,q))^+ > 0$ in $\Phi_1^{sf}(t)$. Thus, algorithm \textsc{PA-EQUI} is decreasing work at speed $s_i(t)$ for at least $\max\{n_e(t)-n_o(t),0\}$ jobs. Hence, the drift $d \Phi^{sf}_1(t)/dt$ with respect to processing by the algorithm \textsc{PA-EQUI} is
\begin{align}\nn
d \Phi^{sf}_1(t)/dt & \stackrel{(a)}= P\left(\frac{n(t)}{N}\right)\left(-\max\{n_e(t)-n_o(t),0\} s_i(t)  \right),\\   \label{eq:phiboundsf3}
  &\stackrel{(b)} = \left(-\max\{n_e(t)-n_o(t),0\}  \right)
\end{align}
where $(a)$ follows since for at least $\max\{n_e(t)-n_o(t),0\}$ jobs, \textsc{PA-EQUI}  is decreasing work at speed $s_i(t)$, while $(b)$ follows since 
$s_i(t)=P\left(\frac{N}{n_e(t)}\right)$ for all jobs $i$ being processed by the \textsc{PA-EQUI} and $n(t) \ge n_e(t)$ by definition.

Combining \eqref{eq:phiboundsf1}, \eqref{eq:phiboundsf2}, and \eqref{eq:phiboundsf3}, the proof is complete.
\end{proof}

\end{document}